\documentclass[11pt]{article}
\usepackage{fullpage,xspace}
\usepackage{amsmath,upgreek}
\usepackage{amsfonts}
\usepackage{amssymb}
\newtheorem{theorem}{Theorem}

\newtheorem{corollary}[theorem]{Corollary}

\newtheorem{lemma}[theorem]{Lemma}

\newtheorem{remark}[theorem]{Remark}

\newcommand{\Xomit}[1]{ }
\newenvironment{proof}[1][Proof]{\textbf{#1.} }{\ \rule{0.5em}{0.5em}}

\newcommand{\eps}{\upvarepsilon}

\begin{document}

\date{}
\title{Online bin packing with cardinality constraints resolved}

\author{J\'anos Balogh \thanks{Department of Applied Informatics, Gyula Juh\'asz Faculty of Education,
     University of Szeged, Hungary. \texttt{balogh@jgypk.u-szeged.hu}} \and J\'ozsef B\'ek\'esi \thanks{Department of Applied Informatics, Gyula Juh\'asz Faculty of Education,
     University of Szeged, Hungary. \texttt{bekesi@jgypk.u-szeged.hu}} \and Gy\"{o}rgy
D\'{o}sa\thanks{Department of Mathematics, University of Pannonia,
Veszprem, Hungary, \texttt{dosagy@almos.vein.hu}. } \and Leah
Epstein\thanks{ Department of Mathematics, University of Haifa,
Haifa, Israel. \texttt{lea@math.haifa.ac.il}. } \and Asaf Levin\thanks{Faculty of Industrial Engineering and Management, The Technion, Haifa, Israel. \texttt{levinas@ie.technion.ac.il.}}}
\maketitle

%

\begin{abstract}
Cardinality constrained bin packing or bin packing with
cardinality constraints is a basic bin packing problem. In the
online version with the parameter $k \geq 2$, items having sizes
in $(0,1]$ associated with them are presented one by one to be
packed into unit capacity bins, such that the capacities of bins
are not exceeded, and no bin receives more than $k$ items. We
resolve the online problem in the sense that we prove a lower
bound of $2$ on the overall asymptotic competitive ratio. This
closes this long standing open problem, since an algorithm of an absolute competitive
ratio $2$ is known. Additionally, we significantly improve the
known lower bounds on the asymptotic competitive ratio for every specific
value of $k$. The novelty of our constructions is based on full adaptivity that creates large gaps between item sizes. Thus, our lower bound inputs do not follow the common practice for online bin packing problems of having a known in advance input consisting of batches for which the algorithm needs to be competitive on every prefix of the input. 
\end{abstract}

%

\section{Introduction}\label{intro}
Bin packing with cardinality constraints (CCBP, also called
cardinality constrained bin packing) is a well-known variant of
bin packing \cite{KSS75,KSS77,KP99,CKP03,Epstein05,EL07afptas}. In
this problem, a parameter $k$ is given. Items of indices
$1,2,\ldots,n$, where item $i$ has a size $s_i \in (0,1]$ are to
be split into subsets called bins, such that the total size of
items packed into each bin is at most $1$, and no bin has more
than $k$ items. In the standard bin packing problem, only the
first condition is required.

CCBP is a special case of vector packing (VP) \cite{GareyGJ76}. In
VP with dimension $d \geq 2$, a set of items, where every item is
a non-zero $d$-dimensional vector whose components are rational
numbers in $[0,1]$, are to be split into subsets (called bins in
this case as well) such that the vector sum of every subset does
not exceed $1$ in any component. Given an input for CCBP, an input
for VP is created as follows. For every item, let the first
component be $s_i$, the second component is $\frac 1k$, and the
remaining components are equal to zero (or to $\frac 1k$).

In this paper we study online algorithms, which receive input
items one by one, and pack each new item irrevocably before the
next item is presented, into an empty (new) bin or non-empty bin.
Such algorithms receive an input as a sequence, while offline
algorithms receive an input as a set. By the definition of CCBP,
an item $i$ can be packed into a non-empty bin $B$ if the packing
is feasible both with respect to the total size of items already
packed into that bin and with respect to the number of packed
items, i.e., the bin contains items of total size at most $1-s_i$
and it contains at most $k-1$ items. An optimal offline algorithm, which
uses a minimum number of bins for packing the items, is denoted by
$OPT$. For an input $L$ and algorithm $A$, we let $A(L)$ denote
the number of bins that $A$ uses to pack $L$. We also use $OPT(L)$
to denote the number of bins that $OPT$ uses for a given input
$L$.  The absolute competitive ratio of an algorithm $A$ is the
supremum ratio over all inputs $L$ between the number of its bins
$A(L)$ and the number of the bins of $OPT$, $OPT(L)$. The
asymptotic approximation ratio is the limit of absolute
approximation ratios $R_K$ when $K$ tends to infinity and $R_K$
takes into account only inputs for which $OPT$ uses at least $K$
bins, that is the asymptotic competitive ratio of $A$ is $$\lim_{K
\rightarrow \infty} \sup_{OPT\geq K} \frac{A(L)}{OPT(L)} \ . $$
The term {\it competitive ratio} is used for online algorithms
instead of {\it approximation ratio} and it is equivalent. In this
paper we mostly deal with the asymptotic competitive ratio, and
also refer to it by the term competitive ratio. When we discuss
the absolute competitive ratio, we use this last term explicitly.

In this paper, we resolve the long standing open problem of online
CCBP, in the sense that we find the best overall asymptotic and
absolute competitive ratios. An algorithm with an asymptotic
competitive ratios of $2$ has been designed by Babel et al.
\cite{BCKK04}, and a similar algorithm was shown to have an
absolute competitive ratio of $2$ \cite{BDE}. However, prior to
this work, all lower bounds were strictly smaller than the best
lower bounds for standard bin packing \cite{Vliet92,BBG}. With the
exception of the case $k=2$ for which simple algorithms have
competitive ratios of $1.5$ \cite{KSS75,Epstein05}, and a more
sophisticated algorithm has a competitive ratio of at most
$1.44721$ \cite{BCKK04}, all lower bounds on the competitive ratio
were implied by partial inputs of inputs used to prove lower
bounds for standard bin packing \cite{Yao80A,Vliet92,BBG} (such
lower bounds can be used for $k \geq \frac {1}{\delta}$ when all
items have sizes no smaller than $\delta$, for a fixed value
$\delta>0$), and modifications of such inputs
\cite{BCKK04,FK13,BDE}. That is, all lower bounds had the form
where a number of lists may be presented, each list has a large
number of items of a certain size (the sequence of sizes of the
different list is increasing, and the numbers of items in the
lists are not necessarily equal). The unknown factor is the number
of presented lists, that is, the input can stop after any of the
lists. See Table \ref{tabtab} for values of previously known lower
bounds.

{\begin{table}[b!] {\small
\renewcommand{\arraystretch}{1.2   }
$$
\begin{array}{||c||c||c||}
\hline \hline

\mbox{Value of \  } k & \mbox{previous lower bound} &  \mbox{new lower bound}  \\

\hline

2 & 1.42764 ~~~ \cite{FK13} &  \frac{10}7\approx 1.42857  \\
\hline
3 &  \frac 32= 1.5 ~~~ \cite{BCKK04}  &   $1.55642$\\
\hline
4 &  \frac 32= 1.5 ~~~  \cite{FK13}&  $1.63330$ \\
\hline
5 & \frac 32=1.5 ~~~ \cite{BDE} & $1.69776$ \\
\hline
6 &  \frac 32=1.5  ~~~ \cite{Yao80A} & $1.74093$ \\
\hline

7 &   \frac{217}{143}\approx 1.51748  ~~~ \cite{BDE} &  $1.77223$ \\
\hline
8 &   \frac{32}{21}\approx 1.52381 ~~~ \cite{BDE} &   $1.79634$ \\
\hline
9 &  \frac{189}{124}\approx 1.524194 ~~~ \cite{BDE}  & $1.81563$\\
\hline
10 &  \frac{235}{154}\approx 1.52597  ~~~ \cite{BDE} & $1.83148$ \\
\hline

200000 &  1.54037  ~~~ \cite{BBG}  & $1.99999$ \\
\hline

k \rightarrow \infty  &   \frac{248}{161}\approx1.54037  ~~~ \cite{BBG} & {\large\boldsymbol{2}}   \\

\hline \hline
\end{array}
$$}
\caption{\label{tabtab} Bounds for $2 \leq k \leq 10$. The middle
column contains the previously best known asymptotic lower bounds on
the asymptotic competitive ratio. The right column contains our
improved lower bounds.}
\end{table}
}

In this work, we take a completely different approach for proving
lower bounds, where many of the item sizes are based on the complete and precise
action of the algorithm up to the time it is presented. While a
similar approach was used for the very limited special case of $k=2$ in the past
\cite{Blitz,BCKK04,FK13}, it was unclear how and if it could be used for
$k>2$. In a nutshell, in these lower bound sequences for $k=2$, sub-inputs were
constructed such that items packed in certain ways (for example,
as the second item of a bin) had much larger sizes than items of
the same sub-input packed in other ways. Here, we generalize the
approach for larger values of $k$ by defining careful
constructions where {\em sufficiently large multiplicative gaps} are
created. This requires much more delicate procedures where item
sizes are defined.

Additionally, we improve the lower bounds for all values of $k$,
and in particular, prove lower bounds above the best known lower
bound on the competitive ratio for standard online bin packing,
$1.54037$ \cite{BBG} for $k \geq 3$. Already for $k=3$ our
lower bound is above $1.55$, and already for $k=4$, our lower
bound is above the competitive ratio of many algorithms for
standard online bin packing (see for example
\cite{RaBrLL89,Seiden02J}).

Our result provides, in particular, a lower bound of $2$ for the
asymptotic competitive ratio of vector packing in two dimensions.
The previously known lower bounds for vector packing are as
follows. The best results for constant dimensions are fairly low,
and tending to $2$ as the dimension $d$ grows to infinity
\cite{GaKeWo94,BlVlWo96,Blitz}, while a lower bound of
$\Omega(d^{1-\eps})$ was given by Azar et al. \cite{ACKS13} for the case where both $d$ and the optimal cost are  functions of a common parameter $n$ that grow to infinity when $n$ grows to infinity, and thus this result does not give any lower bound on the competitive ratio for constant values of $d$  (see
also \cite{ACFR16} for results on vectors with small components).
In particular, the best lower bound for $d=2$ prior to this work
was $1.67117$ \cite{GaKeWo94,BlVlWo96,Blitz}. An upper bound of
$d+0.7$ on the competitive ratio is known \cite{GareyGJ76}.

Note that the offline CCBP problem is NP-hard in the strong sense,
and approximation schemes are known for it \cite{CKP03,EL07afptas}.
We note that for online CCBP, it is sometimes the case that the competitive
ratio for CCBP is larger by $1$ with comparison to that of
standard bin packing \cite{KSS75,JDUGG74,LeeLee85,Epstein05}.
Interestingly, this is not the case here.

\paragraph{Paper outline.} We discuss general properties in Section \ref{prem}, and we define
procedures for constructing sub-inputs in Section \ref{stru}. Our
main result, an overall lower bound of $2$ on the competitive
ratio of any online algorithm for CCBP is proved in Section
\ref{mainres}, and improved lower bounds for fixed values of $k$
are given in Section \ref{fixd}.

\section{Preliminaries}\label{prem}
The analysis of the lower bounds on the asymptotic competitive ratio of online algorithms will be based on the following lemma that basically allows us to disregard a constant number of bins in the costs of the optimal solution and the solution returned by the algorithm.

\begin{lemma}\label{modcr}
Consider an algorithm $ALG$, such that the asymptotic competitive
ratio of the algorithm $ALG$ is at most $R$, where $R \geq 1$ is a
fixed value, and let $f(n)$ denote a positive function such that
$f(n)=o(n)$ and for any input, $ALG(I) \leq R \cdot
OPT(I)+f(OPT(I))$. Let $C_a \geq 0$, $C_b \geq 0$ be constants.
Assume that for a given integer $N_0$, for any integer $n \geq
N_0$ there is an input $I^n$ for which $OPT(I^n)=\Omega(n)$, then
we have $$R \geq \limsup_{n \rightarrow \infty}
\frac{ALG(I^n)+C_a}{OPT(I^n)-C_b} \ . $$
\end{lemma}
\begin{proof}
We have $$\frac{ALG(I^n)+C_a}{n} \leq R \cdot
\frac{OPT(I^n)-C_b}{n}+\frac{C_a+R \cdot
C_b}{n}+\frac{f(OPT(I^n))}{n}$$ for any $n \geq N_0$.

Since $ALG(I^n) + C_a \geq OPT(I^n) -C_b $ and $OPT(I^n)-C_b = \Omega(n)$ while $C_a+R \cdot C_b+f(OPT(I^n))=o(n)$, letting $n$ grow to infinity implies that $R \geq \limsup_{n \rightarrow \infty}
\frac{ALG(I^n)+C_a}{OPT(I^n)-C_b}$.
\end{proof}

\medskip

In what follows, we will use Lemma \ref{modcr} as follows. We
construct inputs whose size depends on a parameter $N$, so that
the costs of optimal solutions increase with the input size. We
will compare the cost of the algorithm plus a suitable
non-negative constant to the optimal cost minus a suitable
non-negative constant by considering their ratio.

\section{Constructions of sub-inputs}\label{stru}
In this section we introduce the core of our lower bound
constructions. In such constructions, we {\em adaptively} present inputs that are
based on the behavior of the algorithm. More specifically, we
define several procedures that construct sub-inputs according to
certain conditions. Similarly to \cite{BCKK04,FK13,Blitz} (and
other work on online problems), a new input item is presented at
each time, where its size is based on the action of the algorithm on previous items.
For example, if the previous item was packed into an empty bin,
the size of the next item is different from the size that would be
used if the previous item is added to a non-empty bin. In order to
ensure that the properties are satisfied, we will define
invariants, and we will prove the specific properties that we need in the sequel via
induction.

In the first procedure, the most important property is that there
will be a gap between two types of items constructed by applying
the procedure, in the sense that the procedure creates items that
will be called small and items that will be called large, any
large item is larger than any small item, and there is a
requirement on the size ratio that will be satisfied (a
multiplicative gap between the size of the smallest large item and
the largest small item). Such constructions differ from previous
work  \cite{BCKK04,FK13,Blitz} where only an additive gap was
created. The gap was always positive, but it could be arbitrarily
small. In particular, one limitation was that it was unknown how
such an approach could be used for $k>2$.

We will also use this method to construct sub-inputs with large
items, such that there is a multiplicative gap in the
differences between $1$ and the items sizes. This new method will
allow us to exploit the major differences in sizes of items of
different kinds to provide a tight overall result and new and
significantly improved lower bounds on the asymptotic competitive
ratio for fixed values of $k$.

\subsection{Procedure SMALL}
In this first procedure called SMALL, a rational value $0<\eps
\leq 1$, and an integer upper bound $N$ on the number of items to be presented are
given. The goal is to present (at most) $N$ items of sizes in
$(0,\eps]$, such that every item will be seen as either a small
item or a large item, and such that any large item is more than
$k$ times larger than any small item. In fact, a stronger
requirement on the item sizes will hold. Moreover, all item sizes
will be rational. Given two logical conditions, $C_1$ and $C_2$
specified for each construction (such that for every packed item,
exactly one of them holds), a new item will be defined as {\em
small} if $C_1$ holds and it will be defined as {\em large} if
$C_2$ holds. There is a third condition $C_3$ that is based on the
packing of the prefix of items introduced so far, and the
sub-input is stopped if $C_3$ holds.

Let $N$ be the upper bound on the number of items that will be created by the procedure. Let $N' \leq N$ be the number of items (where $N$ is known in advance while $N'$ is not necessarily known in advance and it becomes known when $C_3$ holds for the first time).
The item sizes $a_1,a_2,\ldots,a_{N'}$ will be defined based on
another sequence $x_1,x_2,\ldots,x_{N'}$, such that $a_i=\eps \cdot
k^{-x_i}$ for $1 \leq i \leq {N'}$. The values $x_i$ will be integral in order to ensure that the values $a_i$ will be rational. There will also be two sequences of values $\tau_1,\ldots,\tau_{N'}$ and $\rho_1,\ldots,\rho_{N'}$, representing thresholds on item sizes of further items.

Let $\tau_0= 2^{N+2}$, $\rho_0= 2^{N+3}$, and $i=1$. The process
is defined as follows for any given value of $i$ (such that $1
\leq i \leq N'$). Let $x_i=\frac{\tau_{i-1}+\rho_{i-1}}2$ (we will
show that these values are integers). After the algorithm packs
item $i$, if $C_1$ holds, let $\tau_i=\tau_{i-1}$ and $\rho_i=x_i$
and if $C_2$ holds, let $\tau_i=x_i$ and $\rho_i=\rho_{i-1}$. If
$C_3$ holds or $i=N$, stop and otherwise increase $i$ by $1$.

Intuitively, the process is as follows. The interval $(\tau_i,\rho_i)$ contains the $x_j$ values of all further items (with $j>i$), and for $j \leq i$, all items satisfying $C_1$ have $x_j$ values in $[\rho_i,\rho_0)$ and all items satisfying $C_2$ have $x_j$ values in $(\tau_0,\tau_i]$. In each iteration $i$, the new values $\tau_i,\rho_i$ are defined such that these requirements are satisfied. In particular, the
$x_i$ values of any item satisfying $C_1$ are larger than those of items satisfying $C_2$. Next, we establish the invariants of this procedure.

\begin{lemma}
Let $N'$ be the number of items. For any $i$ such that $1 \leq i \leq N'$, $\rho_i \leq \rho_{i-1}$ and $\tau_i \geq \tau_{i-1}$. Additionally, we have $\rho_i-\tau_i = 2^{N+2-i}$, all $x_i$ values are integral, if item $i$ satisfies $C_1$,
$x_i \geq \rho_{N'}$ and otherwise $x_i \leq \tau_{N'}$.
\end{lemma}
\begin{proof}
We start with showing that the $x_i$ values as well as $\rho_i$
and $\tau_i$ are integral and $\rho_i-\tau_i = 2^{N+2-i}$. We
prove this by induction. Indeed $\rho_0=2^{N+3}$ that is integral,
$\tau_0=2^{N+2}$ that is an integer as well.  Furthermore,
$\rho_0-\tau_0= 2^{N+2}$, and $x_1=3 \cdot 2^{N+1}$ that is an
integer, and no matter if the first item satisfies $C_1$ or $C_2$,
we have that both $\rho_1$ and $\tau_1$ are integers, and
$\rho_1-\tau_1= 2^{N+1}$. Thus, the cases $i=0$ and $i=1$ for the
induction claim hold.  Assume that $\rho_{i-1}-\tau_{i-1} =
2^{N+3-i}$ holds for some $i$ where $1 \leq i \leq N'-1$. Then,
$$x_i=\frac{\tau_{i-1}+\rho_{i-1}}2=\tau_{i-1}+\frac{\rho_{i-1}-\tau_{i-1}}2=\tau_{i-1}+2^{N+3-i} \ , $$
which is an integer for $1 \leq i \leq N$, since $\tau_{i-1}$ is
an integer. Moreover, if $\tau_i=\tau_{i-1}$ and $\rho_i=x_i$,
then $\rho_i-\tau_i=x_i-\tau_{i-1}=\frac{\rho_{i-1}-\tau_{i-1}}2$,
and otherwise $\tau_i=x_i$ and $\rho_i=\rho_{i-1}$, then
$\rho_i-\tau_i=\rho_{i-1}-x_i=\frac{\rho_{i-1}-\tau_{i-1}}2$. In
both cases, $\rho_i-\tau_i=2^{N+2-i}$ and both $\tau_i$ and
$\rho_i$ are integers. Since, in particular, for any $i$,
$\rho_i>\tau_i$ holds and $x_{i+1}$ is their average, we find
$\tau_i<x_{i+1}<\rho_i$. Thus, $\rho_i \leq \rho_{i-1}$ and $\tau_i
\geq \tau_{i-1}$ holds for any $i$.

Finally, since in the case that item $i$ satisfies $C_1$, we let $\rho_i=x_i$, and in the case that item $i$ satisfies $C_2$, we let $\tau_i=x_i$, we get $x_i = \rho_i \geq \rho_{i+1} \geq \ldots \geq \rho_{N'}$ in the first case, and $x_i = \tau_i \leq \tau_{i+1} \leq \ldots \leq \tau_{N'}$ in the second case.
\end{proof}

\begin{corollary}
For any item $i$, $a_i \in  \left( \eps \cdot
k^{- 2^{N+3}} , \eps \cdot
k^{-2^{N+2}} \right)$. For any item $i_1$ satisfying $C_1$ and any item $i_2$ satisfying $C_2$, it holds that $\frac{a_{i_2}}{a_{i_1}} > k$.
\end{corollary}
\begin{proof}
The first claim holds by definition. Since we have $x_{i_1} \geq \rho_{N'}$ and $x_{i_2} \leq \tau_{N'}$, we get
$\frac{a_{i_2}}{a_{i_1}} > k^{\rho_{N'}-\tau_{N'}}$, Using
$\rho_{N'}-\tau_{N'} = 2^{N+2-N'} \geq 4$ as $N' \leq N$, we find $\frac{a_{i_2}}{a_{i_1}} \geq k^4 > k$.
\end{proof}

Note that it is possible that the constructed input is such that
there are only items satisfying $C_1$ or only items satisfying
$C_2$.

\subsection{Procedure LARGE}
The second type of input is such that all items have sizes in
$(1-\eps,1)$ for a given value $\eps>0$. The construction is the
same as before, but the size of the $i$th item is $b_i=1-a_i$. The
terms ``small'' and ``large'' refer to difference between the size
of the item and $1$.

\begin{corollary}
All item sizes $b_i$ for $1 \leq i \leq N$ are in $\left( 1-\eps \cdot
k^{-2^{N+2}},1- \eps \cdot k^{-2^{N+3}}\right)$.

The sizes of  any small item $i_s$ and any large item
$i_l$ satisfy $1-b_{i_l}> k \cdot (1-b_{i_s})$.
\end{corollary}

%
%
%

\subsection{Procedure SMALLandLARGE}
We will also use a procedure where the conditions $C_1$
and $C_2$ are not fixed, and they are based on additional
properties of the packing and the input that has been presented so
far. Moreover, in this case the size of each item is based on
$a_i$, but it is fixed for each item separately (it will be either
$a_i$ or $1-a_i$).  In this construction the sub-input will be decomposed into {\em parts} where for an item of an odd indexed part the size of the item will be $1-a_i$, whereas for an item of an even indexed part the size of the item will be $a_i$.  The definitions of $C_1$ and $C_2$ will also depend on the parity of the index of the part containing the item. This procedure is called SMALLandLARGE.


%
%


\section{A lower bound of $\boldsymbol{2}$}\label{mainres}

Let $N$ be a large integer. Apply procedure SMALL with $\eps=1$
for the construction of $N$ items (i.e., condition $C_3$ never happens). The condition $C_2$ is that the item is packed as the first item of some bin (into an empty bin), and the condition $C_1$ is that the item is packed into a non-empty bin.
The item sizes are no larger
than $\frac 1 {k^4}$.
The multiplicative
gap between the smallest large item and the largest small item is larger than $k$.
The $N$ items presented so far will be called the first phase
items. Let $\delta>0$ denote the largest size of any first phase
item packed not as a first item of a bin (the largest small item).
Let $\alpha=k\cdot \delta$. Any first phase item that is packed as
the first item of a bin (a large item) has size strictly above
$\alpha$. Let $\Delta<\frac 1{k^3}$ be the largest size of any
first phase item. Obviously, $1-k\Delta
> 1-\frac{1}{k^2} > \frac 12$.

For the first phase items, let $X_k$ denote the number of bins
packed by the algorithm that contain $k$ items, and let $Y$ denote
the number of other bins (such that there are $X_k+Y$ bins in
total after $N$ items have been presented).

The first phase items are followed by another set of items called
the second phase items. This set of items is selected out of two
possible options. The first option is that $\lceil \frac{N}{k-1}
\rceil$ items of size $1-k\Delta$ arrive, and the second option is
that $\lceil \frac{ N-X_k-Y}{k-1} \rceil$ items of size
$1-\alpha=1-k\delta$ arrive. In both cases it is possible to
create a solution offline such that each bin (except for possibly
two bins) has $k$
items. In the first case, an
offline solution has $\lceil \frac{N}{k-1} \rceil$ bins, each with
one item of size $1-k\Delta$ and an arbitrary subset of $k-1$
first phase items (the last bin may have a smaller number of such
items).  Such a solution is obviously optimal. In the second case, an offline solution has $\lceil
\frac{N-X_k-Y}{k-1} \rceil$ bins, each with one item of size
$1-k\delta$ and $k-1$ small first phase items, and $\lceil
\frac{X_k+Y}{k} \rceil$ bins with $k$ large first phase items (for
each one of these two bin types, the last bin may have a smaller
number of such items).  Indeed the last solution is an optimal solution though we will only use that it is a feasible solution.

In the first case, the algorithm cannot use the bins that already
have $k$ items for packing second phase items, and its cost is at
least $X_k+\lceil \frac{N}{k-1} \rceil \geq X_k+\frac{N}{k-1}$. In
the second case, the algorithm cannot use any of its bins to pack
any second phase item, as each bin has a large first phase item of
size above $\alpha$, so its cost is $$X_k+Y+\left\lceil
\frac{N-X_k-Y}{k-1} \right\rceil \geq X_k+Y+\frac{N-X_k-Y}{k-1} \ . $$

We call the two inputs (of the two cases) $I_1$ and $I_2$.
Obviously, since the input consists of more than $N$ items,
$OPT(I_1) = \Omega(\frac{N}k)$ and $OPT(I_2) = \Omega(\frac{N}k)$
hold. Letting $N=kn$ provides an input $I^n$ as required. By Lemma \ref{modcr}, we will analyze modified competitive ratios of the form $\frac{ALG(I)+C_a}{OPT(I)-C_b}$ for fixed constants $C_a$ and $C_b$.

For the input $I_1$, $OPT(I_1)-1 \leq \frac{N}{k-1}$ and $ALG(I_1) \geq X_k+\frac{N}{k-1}$.
For the input $I_2$, $OPT(I_2)-2 \leq \frac{N-X_k-Y}{k-1}+\frac{X_k+Y}{k}$ and $ALG(I_2) \geq X_k+Y+\frac{N-X_k-Y}{k-1}$.

First, we analyze the competitive ratio $r$ for input $I_2$ and
show that it tends to $2$ as $k$ grows to infinity. Let $Z=X_k+Y$.
We have $OPT(I_2)-2 \leq \frac{N-Z}{k-1}+\frac{Z}{k}$ and
$ALG(I_2) \geq Z+\frac{N-Z}{k-1}$. Thus, $r \geq
\frac{kZ(k-1)+k(N-Z)}{k(N-Z)+(k-1)Z}=\frac{Z(k^2-2k)+kN}{kN-Z}$.
Since $Z \geq \frac{N}k$ and the last lower bound on $r$ is a
ratio between an increasing function of $Z$ and a decreasing
function of $Z$, we conclude that by substituting $\frac Nk$
instead of $Z$ in the last bound, we achieve a valid lower bound
on $r$.  Thus, we have $r \geq \frac{N(k-2)+kN}{kN- \frac
Nk}=\frac{2-2/k}{1-1/(k^2)}=\frac{2k}{k+1}$ and the last bound
tends to $2$ when $k$ grows to infinity. By Lemma \ref{modcr}, the
overall (asymptotic) competitive ratio is at least $2$. Since
there is a $2$-competitive algorithm for any value of $k$
\cite{BCKK04} (even for the absolute competitive ratio
\cite{BDE}), we establish the following.

\begin{theorem}
The overall asymptotic and absolute competitive ratios for bin packing with cardinality constraints are equal to $2$.
\end{theorem}

To obtain a better lower bound on the asymptotic competitive ratio $r$ for a fixed value of $k \geq 3$, we use $I_1$ as well.
By $r \geq \frac{ALG(I_1)}{OPT(I_1)-1}\geq \frac{X_k+N/(k-1)}{N/(k-1)}$ we have $(k-1)X_k \leq (r-1)\cdot N$. By counting arguments, $N \leq kX_k+(k-1)Y$ holds, and we get $X_k \geq N - (k-1)Z$, and $(r-1)N \geq (k-1)X_k \geq (k-1)(N-(k-1)Z)=(k-1)N-(k-1)^2\cdot Z$. Rearranging gives $$Z \geq \frac{(k-r)N}{(k-1)^2} .$$ As we saw earlier, by using $I_2$ we have $r \geq \frac{Z(k^2-2k)+kN}{kN-Z}$, which is equivalent to $$Z (k^2-2k+r) \leq kN(r-1) .$$ Combining the lower bound and upper bound on $Z$ results in $$\frac{(k-r)N(k^2-2k+r)}{(k-1)^2} \leq kN(r-1) , $$ or equivalently
 $$r^2+r(k^3-k^2-2k)-(2k^3-4k^2+k) \geq 0 . $$
 Since $k^3-k^2-2k \geq 0$ holds for $k \geq 2$ and $2k^3-4k^2+k>0$ holds for $k \geq 2$, it is sufficient to find the (unique) positive root which is equal to $$\frac{2k+k^2-k^3+\sqrt{(k^3-k^2-2k)^2+4(2k^3-4k^2+k)}}{2} \ . $$ The last expression is a lower bound on $r$ and thus the following holds.
\begin{theorem}\label{thm6}
For any $k \geq 3$, the asymptotic competitive ratio for bin packing with cardinality constraints is at least $$\frac{2k+k^2-k^3+\sqrt{k^6-2k^5-3k^4+12k^3-12k^2+4k}}2 \ . $$
\end{theorem}


The last lower bound is equal to approximately $1.54983$ for
$k=3$, $1.63330$ for $k=4$, $1.69047$ for $k=5$, $1.73214$ for
$k=6$, $1.76388$ for $k=7$, $1.78888$ for $k=8$, $1.80909$ for
$k=9$, and $1.82575$ for $k=10$. For $k=2$ the resulting lower
bound is $\sqrt{2}$ and the construction (for the case $k=2$) is
indeed similar to that of \cite{Blitz,BCKK04}.

\section{Better lower bounds for some small values of
$\boldsymbol{k}$}\label{fixd}

In this section we prove the next theorem that improves the resulting bounds of Theorem \ref{thm6} for these values of $k$.

\begin{theorem}
The following values are lower bounds on the asymptotic
competitive ratio.
\begin{itemize}
\item Approximately $1.42857$ for $k=2$ (the exact value of this
lower bound is $\frac {10}7$),
\item approximately $1.55642$ for
$k=3$,
\item approximately $1.69776$ for $k=5$,
\item approximately $1.74093$ for $k=6$,
\item approximately $1.77223$ for $k=7$, \item approximately $1.79634$ for $k=8$, \item approximately $1.81563$ for $k=9$, \item and approximately $1.83148$ for $k=10$.
\end{itemize}
\end{theorem}

\subsection{$\boldsymbol{k=3}$}

Given a positive integer $N$, let $\eps = 3^{- 2^{2N+4}}$ (where
$\eps < \frac 1{1000}$). Use SMALL to present exactly $N$ items,
such that $C_1$ is defined as the condition that the new item is
packed into a bin that already has at least one item and $C_2$ is
the condition that the new item is packed into an empty bin.
Recall that all item sizes are no larger than $\eps$. This will be
called the first phase of the input and its items will be called
first phase items.


Let $u_1$ be the size of the largest small item and let $u_2$ be
the size of the smallest large item, where $u_2 > k \cdot u_1$.
Let $X_{\ell}$ denote the number of bins with $\ell$ items packed
by the algorithm in the first phase (for $1 \leq \ell \leq 3$). If
the input is stopped at this time, we call the input $J_1$. As
$OPT(J_1) = \lceil \frac N3 \rceil$, we have $OPT(J_1)-1 \leq
\frac{N}3$ and $ALG(J_1)=X_1+X_2+X_3$. Thus, for the competitive
ratio $r$ of the set of inputs we define for $N$, $X_1+X_2+X_3
\leq r \cdot \frac{N}3$, while $X_{\ell} \geq 0$ holds for
$\ell=1,2,3$, and by $X_1+2X_2+3X_3 = N$, we have $X_1+X_2+X_3
\leq r \cdot \frac{X_1+2X_2+3X_3}3$.

There are three possible continuations to the input, resulting in
inputs $J_2$, $J_3$, and $J_4$. The second phase of $J_2$ consists
of $\lceil \frac N2 \rceil$ (second phase) items, each of size
$1-2\eps>\frac 12$. Since any first phase item has size of at most
$\eps$, and any second phase item requires a separate bin (where
it can be packed with two first phase items possibly with one in
the last such bin, if $N$ is odd), we have $OPT(J_2) -1 \leq \frac
N2$. The algorithm cannot pack any additional item in a bin
containing three items, so $ALG(J_2) \geq X_3 + \lceil \frac N2
\rceil \geq X_3 + \frac N2$. Thus, $X_3 + \frac N2 \leq r \cdot
\frac N2$, or equivalently $X_1 + 2 X_2+5 X_3 \leq r \cdot
(X_1+2X_2+3X_3)$.

The second phase of $J_3$ consists of $\lceil
\frac{N-X_1-X_2-X_3}2 \rceil$ (second phase) items, each of size
$1-2 \cdot u_1 > \frac 12$. One feasible solution would be to pack
$\lceil \frac{X_1+X_2+X_3}3 \rceil$ bins containing (at most)
three large first phase items, and $\lceil \frac{N-X_1-X_2-X_3}2
\rceil$ bins containing one second phase item and (at most) two
small first phase items. This solution is valid since any three
first phase items can be packed into a bin (and the number of
large first phase items is exactly $X_1+X_2+X_3$), and
additionally, any two small first phase items have total size of
at most $2 \cdot u_1$. Thus, $OPT(J_3)-2 \leq
\frac{X_1+X_2+X_3}3+\frac{N-X_1-X_2-X_3} 2 = \frac N2 - \frac
{X_1+X_2+X_3}6$. The algorithm has one large first phase item (of
size above $3u_1$) packed into each bin used for first phase
items. Thus, it packs each second phase item into a new bin. We
get $$ALG(J_3) \geq X_1+X_2+X_3+ \frac{N-X_1-X_2-X_3}2=\frac N2
+\frac{X_1+X_2+X_3}2 \ . $$ Thus, $ \frac N2 + \frac{X_1+X_2+X_3}2
\leq r \cdot (\frac N2 - \frac{X_1+X_2+X_3}6)$, or equivalently
$6X_1+9X_2+12X_3 \leq r \cdot (2X_1+5X_2+8X_3)$.

Next, to define $J_4$, let $\eps'=1$, and use LARGE to define the
second phase consisting of second phase items as follows. The
number of items will be at most $2N$, and the size of each item
will be in $(1-3^{- 2^{2N+2}} , 1-3^{-2^{2N+3}})$. The size of
each item is above $\frac 12$, and it can be packed with any one
or two first phase items. Condition $C_1$ is that the item is
packed into a new bin, $C_2$ is that the item is not packed into a
new bin, and $C_3$ is that the number of items satisfying $C_1$ is
exactly $\lceil \frac N2 \rceil$. Since $X_1+X_2+X_3 \leq
X_1+2X_2+3X_3=N$, the total number of  second phase items is at
most $X_1+X_2+\lceil \frac N2 \rceil \leq 2N$, so the upper bound
on the number of items is sufficiently large and $C_3$ happens before reaching the bound of $2N$. Let $d$ denote the
number of second phase items introduced by the algorithm that
satisfied $C_2$. Since all second phase items have sizes above
$\frac 12$, and a bin that already contains three items cannot
receive additional items, we find $d \leq X_1+X_2$. Let
$0<\zeta_1<\zeta_2<\frac 1{100}$ where $\zeta_2
> 3 \cdot \zeta_1$ be such that every  second phase item that
satisfied $C_1$ has size no smaller than $1-\zeta_1$ and every
item satisfying $C_2$ has size no larger than $1-\zeta_2$. The
third phase consists of $2d$ (third phase) items, each of size
$\frac{\zeta_2}2$ (where $\frac{\zeta_2}2 > \zeta_1$). An optimal
solution packs the items as follows. Each second phase item
satisfying $C_2$ is packed with two third phase items. Each second
phase item satisfying $C_1$ (recall that there are $\lceil \frac
N2 \rceil$ such items) is packed with two first phase items (there
is at most one bin with one first phase item). The solution is
feasible as the size of any second phase item is in $(1-3^{-
2^{2N+2}} , 1-3^{-2^{2N+3}})$, any two first phase items have a
total size of at most $2 \cdot 3^{- 2^{2N+4}} < 3^{- 2^{2N+3}}$,
and those second phase items packed with third phase items have
sizes no larger than $1-\zeta_2$. This solution is optimal since
its cost is equal to the number of second phase items, and their
sizes are above $\frac 12$. Thus, $OPT(J_4) -1 \leq d+\frac N2$.
Note that the number of items in $J_4$ is at least $3\cdot
(d+\frac N2)$.  The algorithm packs every item of the second phase
satisfying $C_1$ into a separate bin that will not receive
additional items (and there are at least $\frac N2$ such bins).
When such an item arrives, it is packed into a new bin, and its
size is above $\frac 12$ and at least $1-\zeta_1$, while any
additional second phase item has a size above $\frac 12$, and any
third phase item has a size above $\zeta_1$. Other items can be
packed in triples in the best case (and there are $3d+N$ such
items). Therefore, $ALG(J_3)\geq \frac{N+3d}3+\frac
N2=d+\frac{5N}6$. We find $$5X_1+10X_2+15X_3+6d \leq
r(3X_1+6X_2+9X_3+6d) \ .$$

Solving the mathematical (quadratic) program (MP) for minimizing $r$ subject to the above constraints gives $r \geq 1.55642$.



\subsection{$\boldsymbol{k\geq 5}$}
Here, we follow the construction of Section \ref{mainres} of introducing either $I_1$ or $I_2$, but we add a third option as follows.
Let $X_{k-1}$ denote the number of bins packed with $k-1$ items after the first phase, and let $Y'$ be the number of bins packed with at most $k-2$ items ($X_k$ will still denote the number of bins packed with exactly $k$ items after the first phase). Thus, $Y=X_{k-1}+Y'$ is the number of bins that can accommodate an additional item. We have $N \leq kX_k+(k-1)X_{k-1}+(k-2)Y'$.
Input $I_3$ is similar to $I_2$ in the sense that its first two phases are the same, but it also has a third phase in which $\lambda_3= 2 \lceil \frac{Y'+X_{k-1}+X_{k}}{k-2} \rceil$ (third phase) items, each of size $0.35$, are presented. The sizes of the first phase items are such that every second phase item has a size above $0.65$, and every $k-2$ first phase items have a total size below $0.3$.

Bins of the solution returned by the algorithm with second phase items are separate bins containing one item each at termination.
Bins with $k$ items (packed with first phase items) cannot receive additional items. A bin with $k-1$ first phase items can receive one third phase item, and a bin with $k-2$ or less first phase items can receive at most two third phase items, due to the sizes of third phase items (which are above $\frac 13$). Thus, at most $X_{k-1}+2Y'$ third phase items can be added to non-empty bins, and the remaining $\max\{0, \lambda_3- X_{k-1}-2Y'\}$ items require new bins.

{\bf Case 1:} Assume that $\lambda_3 \leq  X_{k-1}+2Y'$.  In this
case, the items of the third phase are not presented. For this
case we use a mathematical program with the objective of
minimizing $r$ with non-negativity constraints of all variables,
the condition defining this case, and the constraints $ALG(I_1)
\leq r \cdot (OPT(I_1)-1)$ and $ALG(I_2) \leq r \cdot
(OPT(I_2)-2)$. Letting $\lambda_1=\lceil \frac{N}{k-1} \rceil$ and
$\lambda_2=\lceil \frac{N-X_k-Y}{k-1} \rceil$, the last two
constraints are $X_k+\lambda_1 \leq r \cdot \frac{N+\lambda_1}k$
and $X_k+Y+\lambda_2 \leq r \cdot \frac{N+\lambda_2}k$ (in the
mathematical program we will relax the equality constraints
defining $\lambda_1,\lambda_2,\lambda_3$ into valid inequality
constraints). Thus, the constraints are  the non-negativity
constraints $X_{k-1},X_k, Y',Y, \lambda_1,\lambda_2,\lambda_3 \geq
0$, the constraints defining the auxiliary variables
$Y=X_{k-1}+Y'$, $N \leq kX_k+(k-1)X_{k-1}+(k-2)Y'$, $\lambda_1
\geq \frac{N}{k-1}$, $\lambda_2 \geq \frac{N-X_k-Y}{k-1}$,
$\lambda_3 \geq 2 \cdot \left( \frac{Y'+X_{k-1}+X_{k}}{k-2}
\right)$, the constraint defining this case $\lambda_3 \leq
X_{k-1}+2Y'$, and the last two constraints saying that the
algorithm is $r$-competitive for the two inputs $I_1$ and $I_2$,
that is, $X_k+\lambda_1 \leq r \cdot \frac{N+\lambda_1}k$ and
$X_k+Y+\lambda_2 \leq r \cdot \frac{N+\lambda_2}k$.

{\bf Case 2:} For the other case, i.e., $\lambda_3 \geq
X_{k-1}+2Y'$ we use a similar mathematical program with all these
constraints, with the change that the constraint for $\lambda_3$
is $\lambda_3 \geq  X_{k-1}+2Y'$ (instead of $\lambda_3 \leq
X_{k-1}+2Y'$ in the mathematical program for case 1), and that a
constraint for $I_3$ is added, as we get another constraint from
$ALG(I_3) \leq r \cdot (OPT(I_3)-2)$.

The new constraint for $I_3$ is based on the following reasoning. The
offline solution we consider packs the small first phase items and
the second phase items as before, while every $k-2$ large first
phase items are packed with two third phase items (one bin may
contain a smaller number of large first phase items). Thus,
$OPT(I_3)-2 \leq \frac{N+\lambda_2+\lambda_3}{k}$. The algorithm
has $X_k+X_{k-1}+Y'$ bins created for the first phase items,
$\lambda_2$ bins created for the second phase items, and at least
$\lceil \frac{\lambda_3-X_{k-1}-2Y'}2 \rceil$ bins for the third
phase items (which could not be packed into non-empty bins, as
only $X_{k-1}+2Y'$ third phase items can be packed there, and the
remaining third phase items can be packed in pairs in the best
case). Thus, $$ALG(I_3) \geq X_k+X_{k-1}+Y'+\lambda_2+
\frac{\lambda_3-X_{k-1}-2Y'}2=X_k+\frac{X_{k-1}}2+\lambda_2+\frac{\lambda_3}2
\ .$$ Thus, the constraint that we add to the mathematical
program corresponding to the input $I_3$ is
$X_k+\frac{X_{k-1}}2+\lambda_2+\frac{\lambda_3}2 \leq r \cdot
\frac{N+\lambda_2+\lambda_3}{k}$.

The lower bound on $r$ results from the minimum of the solutions of the two mathematical programs stated above. This gives the following lower bounds on the competitive ratio:  approximately $1.69776$ for $k=5$, $1.74093$ for $k=6$, $1.77223$ for $k=7$, $1.79634$ for $k=8$, $1.81563$ for $k=9$, and $1.83148$ for $k=10$.


\begin{remark}
  Another way to define the third phase is to present items of sizes slightly above $\frac 14$ ($3 \lceil \frac{Y+X_{k}}{k-3} \rceil$ such items) instead of items of sizes slightly above $\frac 13$. Analyzing this additional scenario results in slightly better results for $k=8,9,10$ (more specifically, a lower bound of approximately $1.79758$ on the competitive ratio for $k=8$, $1.81812$ for $k=9$, and $1.83442$ for $k=10$). Since the improvement is not significant we omit the details.
\end{remark}

\subsection{$\boldsymbol{k=2}$}


In this section we prove an improved lower bound of $\frac{10}7\approx 1.42857$ for $k=2$. 
In previous constructions procedures of the type presented in
Section \ref{stru} were used. In \cite{Blitz,BCKK04}, items are
presented in just two phases, that is, only one such construction
is used. In \cite{FK13}, there are additional phases, but they are
different from what we present below in the sense that their
number is fixed in advance. Our construction allows a shorter
analysis avoiding linear programs.

Given a positive integer $N$, let $\eps = 2^{- 2^{5N+4}}$ (where
$\eps < \frac 1{1000}$). Use SMALL to present exactly $N$ items,
such that $C_1$ is defined as the condition that the new item is
packed into a bin that already has an item and $C_2$ is the
condition that the new item is packed into an empty bin. Recall
that all item sizes are no larger than $\eps$. This will be called
the first phase of the input and its items will be called first
phase items.

Let $u_1$ be the size of the largest small item and let $u_2$ be
the size of the smallest large item, where $u_2 > u_1$. Let
$u=\frac{u_1+u_2}2$. First phase items that satisfied $C_1$ will
be called $t^-$ items, and first phase items that satisfied $C_2$ will be
called $t^+$-items. Any $t^-$-item has size of at most $u_1<u$ and
any $t^+$-item has size of at least $u_2>u$. Let $X_{\ell}$ denote
the number of bins with $\ell$ items packed by the algorithm in
the first phase (for $\ell =1,2$). The input continues with $X_2$
(second phase) items, each of size $1-u$ (where $1-u > \frac 12$).
Since every bin of the algorithm has a $t^+$-item, the algorithm
will pack every item of size $1-u$ into a separate new bin.  If
the input is stopped at this time, we call the input $J_1$. For an
offline solution, it is possible to pack every $t^-$-item into a
bin together with an item of size $1-u$, and $t^+$-items are
packed in pairs into bins (one bin may contain only one such
item). Thus, $OPT(J_1) \leq X_2 + \lceil \frac {X_1+X_2}2 \rceil$,
we have $OPT(J_1)-1 \leq \frac{X_1+3X_2}2$. We also have
$ALG(J_1)=X_1+2X_2$. Thus, for the competitive ratio $r$ of the
set of inputs we define for $N$, $X_1+2X_2 \leq r \cdot
\frac{X_1+3X_2}2$. Due to the number of items introduced so far,
$X_1+2X_2= N$, so we get $2N \leq r \cdot (N+X_2)$. If for an
infinite number of values $N$, it holds that $X_2 \leq
\frac{2N}5$, we find $2N \leq r \cdot (N+X_2) \leq r \cdot 1.4N$,
proving $r \geq \frac{2}{1.4}=\frac{10}7$. Thus, we focus on the
case $X_2 \geq \frac{2N}5$ (and therefore $X_1=N-2X_2 \leq \frac
N5$).

The third phase will consist of at most $2N+1$ parts. The total
number of items that will be presented (during all  parts of this
phase) is at most $5N$, and they will be presented using procedure
SMALLandLARGE. The procedure is applied once with $\eps'=1$. For
any $i$ such that $a_i$ will be defined, we will have $\eps <
2^{-2^{5N+3}}< a_i < 2^{-2^{5N+2}}< \frac 12$, while any first
phase item has size at most $\eps$. Thus, an item of size $1-a_i$
and a first phase item have a total size of at most
$\eps+1-a_i<1$.

The conditions $C_1$ and $C_2$ are as follows. In parts of the
phase of even indices, $C_1$ is defined as the condition that the
new item is packed into a bin that already has an item and $C_2$
is the condition that the new item is packed into an empty bin. In
parts of the phase of odd indices, $C_2$ is defined as the
condition that the new item is packed into a bin that already has
an item and $C_1$ is the condition that the new item is packed
into an empty bin. Moreover, if an item of index $i$ (in the
phase) is presented in part $t$ of the phase, its size is
$\mu_i=a_i$ if $t$ is even and its size is $\mu_i=1-a_i$ if $t$ is
odd. Note that any item of an odd indexed part of the phase has
size above $\frac 12$, and any item of an even indexed part has
size below $\frac 12$. We will call the items of even indexed
parts that satisfied $C_1$:  $s^-$-items, and we call the items of
even indexed parts that satisfied $C_2$: $s^+$-items. We will call
the items of odd indexed parts that satisfied $C_2$: $b^-$-items,
and we call the  items of odd indexed parts that satisfied
$C_1$: $b^+$-items. For a part of index $z$, let $\nu_z$ denote
the required number of $b^+$-items of this part, if $z$ is odd,
and let $\nu_z$ denote the required number of $s^+$-items of this
part, if $z$ is even. Let $\nu'_z$ denote the number $b^-$-items
of part $z$, if $z$ is odd, and let $\nu'_z$ denote the number of
$s^-$-items of this part, if $z$ is even. The values $\nu'_z$
result from the action of the algorithm in part $z$, while the
values $\nu_z$ are defined based on its action in earlier phases and
earlier parts as follows. For any $z$, let $\nu_z=X_2$, if $z=1$, and
otherwise $\nu_z=\nu'_{z-1}$. If $\nu_{z}=0$, no items are
presented, and the input was stopped already.

Part $z$ lasts until there are $\nu_z$ items that are $b^+$-items,
if $z$ is odd, or until there are $\nu_z$ items that are
$s^+$-items, if $z$ is even. Note that these are exactly the items
that are packed into new bins. If at some time the number of third
phase items reaches the value $5N$, the input is stopped. We will
later show that this never happens, and $\nu_z$ items that are
either $b^+$-items or $s^+$-items are presented as required. The
value $\nu'_z$ is defined to be the number of other items
presented during this part, which are exactly the items that are
packed by the algorithm as second items into non-empty bins. If
$\nu'_z=0$, the input is stopped, and otherwise, the construction
continues to the next part. We will also show that the input is
stopped after at most $2N+1$ parts.

%
%
%
%


For the input $J_2$, consisting of all three phases, we use the
following variables. Let $x=2X_2$ and $y=X_1$. Given the output of
the algorithm, let $z$ denote the number of bins where the first
item is a third phase item (by definition, this item is a
$b^+$-item or a $s^+$-item). Out of these bins, let $z_1$ be the
number of bins with one $b^+$-item, let $z_2$ be the number of
bins with two items (these must be an $s^+$-item and an
$s^-$-item), and let $z_3$ be the number of bins
with one $s^+$-item. 
Out of the
bins with a $t^+$-item that do not have a $t^-$-item, let $y_1$
denote the number of bins that have a $b^-$-item, let $y_2$ denote
the number of bins that have an $s^-$-item, and let $y_3$ be the
number of bins that do not have another item.

By the action of SMALLandLARGE, the value $a_i$ is larger than any
value $a_j$ where $j<i$, if the item corresponding to $a_j$ is an
$s^-$-item or a $b^+$-item, and $a_i<a_j$, if the item
corresponding to $a_j$ is an $s^+$-item or a $b^-$-item.
Therefore, we find that any $b^+$-item can be packed with any
$t^+$-item or with any $s^-$-item of an earlier part of phase $3$,
and any $b^-$-item can be packed with any $s^+$-item of a later
part. However, a $b^+$-item cannot be packed with a $s^-$-item of
a later part. Moreover, a $b^-$-item of part $z$ cannot be packed
with any $s^+$-item of an earlier part. Thus, $b^-$-items can be
packed only into bins already containing $t^+$-items (they also
cannot be added to bins containing an item of size above $\frac
12$ and they cannot be packed into empty bins since in such a case
the item is defined to be a $b^+$-item). Therefore, the total
number of $b^-$-items in all parts cannot exceed $y \leq N$.

We show that the third phase can indeed be constructed. First, we
show that the process can be completed assuming that the upper
bound on the number of items is sufficiently large. Since for part
$z$, the required $\nu_z$ items are items that are packed into new
bins, and since at each time the number of non-empty bins is
finite, these items will be presented eventually. There can be at
most $N$ parts of odd indices in which at least one $b^-$-item is
presented. Thus, the construction must end after $N+1$ parts of
odd indices, if it is not stopped earlier, and the number of parts
is at most $2N+1$.

The number of $t^+$-items is $X_1+X_2=\frac x2+y$. The number of
$t^-$-items is $X_2=\frac x2$. The number of second phase items is
$\frac x2$. The number of $b^+$-items is equal to the number of
$t^+$-items plus the number of $s^-$-items, and it is equal to
$z_1=(\frac x2+ y)+y_2+z_2$. The number of $b^-$-items is equal to
the number of $s^+$-items, and it is equal to $y_1=z_2+z_3$. Thus,
the total number of items in phase $3$ is at most
$$y_1+z_1+(z_2+z_3)+(y_2+z_2)=2y_1+z_1+y_2+z_2$$ $$=2y_1+\frac
x2+y+y_2+z_2+y_2+z_2 \leq 3y +\frac x2 + 2(y_1-z_3) \leq 5y+\frac
x2 \leq 5N \ .$$

It is possible to create an offline solution of cost $\frac
x2+y_1+z_1$ due to the following. Any $t^-$-item can be packed
with any second phase item, and packing such pairs creates
$X_1=\frac x2$ bins. Each $t^+$-item is packed with a $b^+$-item
of the first part of the third phase (the numbers of such items
are equal to $\nu_1$), each $s^-$-item of part $z$ is packed with
a $b^+$-item of part $z+1$ (the number of such items are
$\nu'_z=\nu_{z+1}$) and each $b^-$-item of part $z$ is packed with
a $s^+$-item of part $z+1$ (the number of such items are
$\nu'_z=\nu_{z+1}$). Thus, the number of additional bins is
exactly the number of $b^+$-items and $b^-$-items, which is
$y_1+z_1$. Thus, $OPT(J_2) \leq \frac x2+y_1+z_1$.

As for the algorithm, it can never use bins with $b^+$-items to
pack other items (note that these items are packed first into
bins, by definition), so there are no additional bins except for
the eight types described above. Thus $ALG(J_2) = x+y+z$.

We get $$r-1 \geq \frac{ALG(J_2)}{OPT(J_2)}-1 \geq
\frac{x+y+z}{\frac x2 +y_1+z_1}-1=
\frac{x/2+y_2+y_3+z_2+z_3}{\frac x2 +y_1+z_1}$$ $$=\frac{\frac
x2+y_2+y_3+y_1}{\frac x2 +y_1+\frac x2+ y+y_2+y_1-z_3} \geq
\frac{\frac x2+y}{x +3y} =\frac{x+2y}{2x+6y}=\frac
13+\frac{x}{6(x+3y)}\ . $$ Recall that $y \leq \frac N5$ and $x
\geq \frac{4N}5$. Thus $x>0$, and we will bound
$\frac{x}{x+3y}=\frac1{1+3y/x}$. We get $\frac yx \leq \frac 14$,
so $\frac1{1+3y/x}\geq \frac 47$, and $r-1 \geq \frac 37$, as
required.

\section{Conclusion}
As explained in the introduction, our results imply, in particular
a lower bound of $2$ on the competitive ratio for vector packing
in two or more dimensions. We can slightly improve this result
showing lower bounds above $2$ for vector packing with constant
dimensions (already for two dimensions).


%
%

\bibliographystyle{abbrv}

\begin{thebibliography}{10}


\bibitem{ACFR16}
Y.~Azar, I.~R.~Cohen, A.~Fiat, and A.~Roytman.
\newblock Packing small vectors.
\newblock {\em Proc. of the 27th Annual {ACM-SIAM} Symposium on
Discrete  Algorithms, (SODA'16)}, pages 1511--1525, 2016.



\bibitem{ACKS13}
Y.~Azar, I.~R.~Cohen, S.~Kamara, and F.~B.~Shepherd.
\newblock Tight bounds for online vector bin packing.
\newblock In {\em Proc. of the
45th ACM Symposium on Theory of Computing (STOC'13)}, pages
961--970, 2013.

\bibitem{BCKK04}
L.~Babel, B.~Chen, H.~Kellerer, and V.~Kotov.
\newblock Algorithms for on-line bin-packing problems with cardinality
  constraints.
\newblock {\em Discrete Applied Mathematics}, 143(1-3):238--251, 2004.

\bibitem{BBG}
J.~Balogh, J.~B{\'e}k{\'e}si, and G.~Galambos.
\newblock New lower bounds for certain classes of bin packing algorithms.
\newblock {\em Theoretical Computer Science}, 440-441:1--13, 2012.

\bibitem{BDE}
J.~B{\'e}k{\'e}si,  Gy. D\'osa, and  L. Epstein. \newblock
Bounds for online bin packing with cardinality constraints.
\newblock {\em Information and Computation} 249:190--204, 2016.


\bibitem{Blitz}
D.~Blitz. Lower bounds on the asymptotic worst-case ratios of
on-line bin packing algorithms. M.Sc. thesis, University of
Rotterdam, number 114682, 1996.

\bibitem{BlVlWo96}
D.~Blitz, A.~van Vliet, and G.~J.~Woeginger.
\newblock Lower bounds on the asymptotic worst-case ratio of online bin packing
  algorithms.
\newblock Unpublished manuscript, 1996.

\bibitem{CKP03}
A.~Caprara, H.~Kellerer, and U.~Pferschy.
\newblock Approximation schemes for ordered vector packing problems.
\newblock {\em Naval Research Logistics}, 50(1):58--69, 2003.


\bibitem{Epstein05}
L.~Epstein.
\newblock Online bin packing with cardinality constraints.
\newblock {\em SIAM Journal on Discrete Mathematics}, 20(4):1015--1030, 2006.

\bibitem{EL07afptas}
L.~Epstein and A.~Levin.
\newblock {AFPTAS} results for common variants of bin packing: A new method for
  handling the small items.
\newblock {\em SIAM Journal on Optimization}, 20(6):3121--3145, 2010.

\bibitem{FK13}
H.~Fujiwara and K.~Kobayashi.
\newblock Improved lower bounds for the online bin packing problem with
  cardinality constraints.
\newblock {\em Journal of Combinatorial Optimization}, 29(1):67--87, 2015.

\bibitem{GaKeWo94}
G.~Galambos, H.~Kellerer, and G.~J.~Woeginger.
\newblock A lower bound for online vector packing algorithms.
\newblock {\em Acta Cybernetica}, 10:23--34, 1994.

\bibitem{GareyGJ76}
M.~R. Garey, R.~L. Graham, D.~S. Johnson, and A.~C.-C. Yao.
\newblock Resource constrained scheduling as generalized bin packing.
\newblock {\em Journal of Combinatorial Theory Series A}, 21(3):257--298, 1976.

\bibitem{JDUGG74}
D.~S. Johnson, A.~Demers, J.~D. Ullman, M.~R. Garey, and R.~L.
Graham.
\newblock Worst-case performance bounds for simple one-dimensional packing
  algorithms.
\newblock {\em SIAM Journal on Computing}, 3:256--278, 1974.

\bibitem{KP99}
H.~Kellerer and U.~Pferschy.
\newblock Cardinality constrained bin-packing problems.
\newblock {\em Annals of Operations Research}, 92:335--348, 1999.

\bibitem{KSS75}
K.~L. Krause, V.~Y. Shen, and H.~D. Schwetman.
\newblock {Analysis} of several task-scheduling algorithms for a model of
  multiprogramming computer systems.
\newblock {\em Journal of the ACM}, 22(4):522--550, 1975.

\bibitem{KSS77}
K.~L. Krause, V.~Y. Shen, and H.~D. Schwetman.
\newblock Errata: ``{Analysis} of several task-scheduling algorithms for a
  model of multiprogramming computer systems''.
\newblock {\em Journal of the ACM}, 24(3):527--527, 1977.
%
\bibitem{LeeLee85}
C.~C. Lee and D.~T. Lee.
\newblock A simple online bin packing algorithm.
\newblock {\em Journal of the ACM}, 32(3):562--572, 1985.

\bibitem{RaBrLL89}
P.~Ramanan, D.~J. Brown, C.~C. Lee, and D.~T. Lee.
\newblock Online bin packing in linear time.
\newblock {\em Journal of Algorithms}, 10:305--326, 1989.
%
\bibitem{Seiden02J}
S.~S. Seiden.
\newblock {On the online bin packing problem}.
\newblock {\em Journal of the ACM}, 49(5):640--671, 2002.
\bibitem{Vliet92}
A.~van Vliet.
\newblock An improved lower bound for online bin packing algorithms.
\newblock {\em Information Processing Letters}, 43(5):277--284, 1992.


\bibitem{Yao80A}
A.~C.~C. Yao.
\newblock New algorithms for bin packing.
\newblock {\em Journal of the {ACM}}, 27(2):207--227, 1980.


\end{thebibliography}

\end{document}